\documentclass[twoside,leqno,twocolumn]{article}  
\usepackage{ltexpprt} 
\usepackage{amssymb}  
\usepackage{amsmath}
\usepackage[table]{xcolor}

\title{\Large The Polyhedron-Hitting Problem\thanks{This research 
was partially supported by EPSRC.}}
\author{Ventsislav Chonev \\ 
Department of Computer Science \\
Oxford University, UK
\and 
Jo\"el Ouaknine \\ 
Department of Computer Science \\
Oxford University, UK
\and
James Worrell \\
Department of Computer Science \\
Oxford University, UK
}
\date{}

\begin{document}
\bibliographystyle{plain}

\maketitle
\begin{abstract}
We consider polyhedral versions of Kannan and Lipton's Orbit
Problem~\cite{KL80,orbit}---determining whether a target polyhedron
$V$ may be reached from a starting point $x$ under repeated
applications of a linear transformation $A$ in an ambient vector space
$\mathbb{Q}^m$.  In the context of program verification, very similar
reachability questions were also considered and left open by Lee and
Yannakakis in~\cite{LY92}, and by Braverman in~\cite{Bra06}.  We
present what amounts to a complete characterisation of the
decidability landscape for the Polyhedron-Hitting Problem, expressed
as a function of the dimension $m$ of the ambient space, together with
the dimension of the polyhedral target $V$: more precisely, for each
pair of dimensions, we either establish decidability, or show hardness
for longstanding number-theoretic open problems.
\end{abstract}

\newcommand{\ein}[1]{2^{#1^{O(1)}}} \newcommand{\pin}[1]{#1^{O(1)}}
\newcommand{\ra}{\mathbb{R}\cap\mathbb{A}}
\newcommand{\re}{\mathit{Re}} \newcommand{\im}{\mathit{Im}}
\newcommand{\ess}{\mathcal{S}} \newcommand{\spa}{\mathit{span}}

\setcounter{page}{1}

\section{Introduction}

Given a linear transformation $A$ over the vector space
$\mathbb{Q}^m$, together with a starting point $x$, the \textbf{orbit}
of $x$ under $A$ is the infinite sequence $\langle x, Ax, A^2x,
\ldots, A^jx, \ldots \rangle$. A natural decision problem in
discrete linear dynamical systems is whether the orbit of $x$ ever
hits a particular target set $V$.

An early instance of this problem was raised by Harrison in
1969~\cite{Har69} for the special case in which $V$ is simply a point
in $\mathbb{Q}^m$. Decidability remained open for over ten years, and
was finally settled in a seminal paper of Kannan and Lipton, who
moreover gave a polynomial-time decision procedure~\cite{KL80}. In
subsequent work~\cite{orbit}, Kannan and Lipton noted that the Orbit
Problem becomes considerably harder when the target $V$ is replaced by
a subspace of $\mathbb{Q}^m$: indeed, if $V$ has dimension $m-1$, the
problem is equivalent to the \emph{Skolem Problem}, known to be
\textbf{NP}-hard but whose decidability has remained open for over 80
years~\cite{SKOLEMREF}. Nevertheless, Kannan and Lipton speculated
in~\cite{orbit} that instances of the Orbit Problem with
low-dimensional subspaces as target might be decidable. This was
finally substantiated in~\cite{COW13}, which showed decidability for
vector-space targets of dimension at most 3, with polynomial-time
complexity for one-dimensional targets, and complexity in
$\mathbf{NP}^\mathbf{RP}$ for two- and three-dimensional targets.

In this paper, we study a natural generalisation of the Orbit Problem,
which we call the \textbf{Polyhedron-Hitting Problem}, in which the
target $V$ is allowed to be an arbitrary (bounded or unbounded)
polyhedron.\footnote{This problem was also considered in~\cite{TV11}
  under the appellation of \emph{Chamber-Hitting Problem}. However
  that paper focused on connections with formal language theory
  rather than on establishing decidability.} We present what amounts
to a complete characterisation of the decidability landscape for this
problem, expressed as a function of the dimension $m$ of the ambient
space $\mathbb{Q}^m$, together with the dimension $k$ of the
polyhedral target $V$; more precisely, for each pair of dimensions, we
either establish decidability, or show hardness for longstanding
number-theoretic open problems. Our results are summarised in
Fig.~\ref{fig-main-results}. As our algorithms rely on symbolic
manipulation of algebraic numbers of unbounded degree and height, all
decidable instances lie in $\mathbf{PSPACE}$.

\begin{figure*}
\label{fig-main-results}
\newcommand{\tick}{$\mathbf{PSPACE}$}
\newcommand{\tickp}{$\mathbf{P}$}

\begin{center}
\begin{tabular}{|c|c|c|c|c|c|c|}
\hline
& $m=1$ & $m=2$ & $m=3$ & $m=4$ & $m=k$ & $m\geq k+1$ \\
\cline{1-7}
$k=0$ & \tickp & \tickp & \tickp & \tickp & \tickp & \tickp \\
\cline{1-7}
$k=1$ & \tick & \tick & \tick & \tick & \tick & \tick \\
\cline{1-1}\cline{2-7}
$k=2$ & \cellcolor{gray} & \tick & \tick & \tick & \tick & \tick \\
\cline{1-1}\cline{2-7}
$k=3$ & \cellcolor{gray} & \cellcolor{gray} & \tick & $S_{5}$ & \tick & $S_{5}$ \\
\cline{1-1}\cline{2-7}
$k=4$ & \cellcolor{gray} & \cellcolor{gray} & \cellcolor{gray} & $D$ & $D$ & $D\ \&\ S_5$ \\
\cline{1-1}\cline{2-7}
$k\geq5$ & \cellcolor{gray} & \cellcolor{gray} & \cellcolor{gray} & \cellcolor{gray} & $D$ & $D\ \&\ S_{k+1}$ \\
\hline
\end{tabular}
\end{center}

\caption{Upper and lower complexity bounds for instances of the
  Polyhedron-Hitting Problem in ambient dimension $m$ with a
  $k$-dimensional target. The row $k=0$ corresponds to Kannan and
  Lipton's Orbit Problem~\cite{KL80,orbit}.  Upper complexity bounds
  are denoted by \textbf{P} and \textbf{PSPACE}, indicating membership
  in these classes, whereas lower bounds are denoted by $D$
  (indicating reduction from certain Diophantine-approximation
  problems detailed precisely in Sec.~\ref{subsec: PHPhard}) and $S_d$
  (indicating reduction from Skolem's Problem of order $d$, defined in
  Sec.~\ref{subsec: PHPhard}).}
\end{figure*}

A key motivation for studying the Polyhedron-Hitting Problem comes
from the area of program verification, and in particular the problem
of determining whether a simple while loop with linear (or affine)
assignments and guards will terminate or not. Very similar reachability
questions were considered and left open by Lee and Yannakakis
in~\cite{LY92} for what they termed ``real affine transition
systems''. Similarly, decidability for the special case of the
Polyhedron-Hitting Problem in which the polyhedral target consists of
a \emph{single} halfspace (rather than an intersection of several
halfspaces) was mentioned as an open problem by Braverman in Sec.~6
of~\cite{Bra06}.

It should be noted, however, that the problem considered in the
present paper differs in one fundamental respect from what is
traditionally termed the `Termination Problem' in the program
verification literature (see, e.g.,~\cite{BG14}). The latter studies
termination of while loops for \emph{all} possible initial starting
points (valuations of the variables), rather than for a \emph{fixed}
starting point as we consider in this paper. This distinction
drastically transforms the nature of the problem at hand.

In~\cite{OPW14}, the traditional Termination Problem is solved over
the integers for while loops under certain restrictions (chiefly,
diagonalisability of the associated linear transformation). That paper
relies on markedly different techniques from the present one,
eschewing Baker's Theorem and relying instead on non-constructive
lower bounds on sums of $S$-units (which in turn follow from deep
results in Diophantine approximation), as well as real algebraic
geometry.

The present paper vastly extends our earlier results
from~\cite{COW13}, in which only vector-space targets were
considered. Polyhedra, defined as intersections of (affine)
halfspaces, pose substantial new challenges, as evidenced among others
by the Diophantine-approximation lower bounds that arise for
polyhedral targets of dimension 4 or greater. In addition to classical
tools from algebraic and transcendental number theory such as
Baker's Theorem, the present paper relies crucially on
several tools not invoked in~\cite{COW13}
or~\cite{OPW14}, including techniques from Diophantine approximation,
convex geometry, as well as decision procedures for the existential
fragment of the first-order theory of the reals.

In terms of future work, either establishing complexity lower bounds,
or improving the $\mathbf{PSPACE}$ membership of the decidable problem
instances, stand out as challenging open questions.

\section{Polyhedron-Hitting Problem}

The focus of this paper is the \textbf{Polyhedron-Hitting Problem}: given a
square matrix $A\in\mathbb{Q}^{m\times m}$, a vector $x\in\mathbb{Q}^m$ and
polyhedron $P$ (represented as the intersection of halfspaces),
determine whether there exists a natural number $n$ such that
$A^nx\in P$.  We will denote by $\mathit{PHP}(m,k)$ the version
of the problem in which the ambient space is $\mathbb{Q}^m$ and the target
polyhedron has dimension $k \leq m$.

We begin this section with our decidability results for low-dimensional
versions of the problem. We define two related problems to which we reduce the
Polyhedron-Hitting Problem in order to obtain our complexity upper bounds: the
Extended Orbit Problem and the Simultaneous Positivity Problem. Then we proceed
to give hardness results for higher-dimensional cases by reducing from Skolem's
Problem and from Diophantine approximation.

\subsection{Decidability results}

Our effectiveness result on the Polyhedron-Hitting Problem is the following: 
\begin{theorem}\label{thm: PHPdecidability}
If $k\leq2$ or $m=k=3$, then $\mathit{PHP}(m, k)$ is in $\mathbf{PSPACE}$.
\end{theorem}
The strategy for $\mathit{PHP}(m, k)$ when $k\leq2$ is to reduce to the related
\textbf{Extended Orbit Problem}: given a linear transformation
$A\in\mathbb{Q}^{m\times m}$, a vector $x\in\mathbb{Q}^m$, a target
$\mathbb{Q}$-vector space $V$ defined by a basis $\{ y_1, \dots, y_d
\}\subseteq\mathbb{Q}^m$ and a constraint matrix $B\in(\ra)^{k\times d}$,
determine whether there exists some exponent $n\in\mathbb{N}$ such that
$A^nx\in V$ and the coordinates $u=(u_1,\dots,u_d)^T$ of $A^nx$ with respect to
the basis $\{y_1,\dots,y_d\}$ satisfy $Bu\geq0$. This problem essentially
specialises the target of $\mathit{PHP}$ to a cone and assumes a particular
parametric representation.

We focus first on $\mathit{PHP}(m, 1)$. By Lemma \ref{lem: poly1} in
Appendix \ref{sec: polyhedra},
a one-dimensional polyhedron is of the form
\[ P = \{ v_1 + \alpha v_2 : \alpha\in I \} \]
where $I$ is one of $\mathbb{R}$, $[0,1]$ and $[0,\infty)$.  Moreover, this
parametric representation is computable in polynomial time from the halfspace
description of $P$. Now suppose we wish to find $n\in\mathbb{N}$ and $u_1,u_2
\in\mathbb{Q}$ such that 
\[
\left[\begin{array}{cc}
	A & 0 \\
	0 & 1 \\
\end{array}\right]^n
\left[\begin{array}{c}
x \\ 1
  \end{array}\right]
=
u_1\left[\begin{array}{c} v_1 \\ 1 \end{array}\right]
+
u_2\left[\begin{array}{c} v_2 \\ 0 \end{array}\right]
\]
The $(m+1)$-th component forces any witness to this
problem instance to have $u_1=1$. Therefore, requiring $u_2\geq0$ and
$u_1-u_2 \geq 0$ gives an Extended Orbit instance with a
two-dimensional target space which is positive if
and only if the segment $\{v_1+u_2v_2:u_2\in[0,1]\}$
intersects the orbit $\{A^nx:n\in\mathbb{N}\}$. Requiring instead only
$u_2\geq0$ gives the half-line $\{v_1+u_2v_2: u_2
\in[0,\infty)\}$, whereas setting no restriction 
gives the whole line $\{v_1+u_2v_2:u_2\in\mathbb{R}\}$.
In all cases, the resulting Extended Orbit instance has target space of
dimension two, so by Theorem \ref{thm: EOdecidability} in Section \ref{sec: EO},
$\mathit{PHP}(m, 1)$ is in $\mathbf{PSPACE}$.

Now we move to $\mathit{PHP}(m, 2)$. By Lemma \ref{lem: poly2} in
Appendix \ref{sec: polyhedra}, any
two-dimensional polyhedron can be decomposed into a finite union of simple
shapes: $P=\bigcup_{i=1}^s S_i$ where
\[ 
S_i = 
\{ v_{i_1} + \alpha v_{i_2} + \beta v_{i_3} : 
\alpha\geq0\mbox{ and }\beta\geq0
\mbox{ and }T(\alpha, \beta) \} 
\]
where the predicate $T$ is either $\alpha+\beta\leq1$, or $\beta\leq1$ or
$\mbox{true}$. In fact, it is easy to see from the proof of Lemma \ref{lem:
poly2} that $s\in\ein{\Vert P\Vert}$. For each $i$, the problem of whether
there exists $n$ such that $A^n x\in S_i$ reduces to the Extended Orbit Problem
with a three-dimensional target. For instance, if the predicate $T_i$ is 
$\alpha+\beta\leq1$, that is, $S_i$ is a triangle, then
$A^nx\in S_i$ if and only if there exist $u_1,u_2,u_3\in\mathbb{Q}$ such that
$u_2\geq0$, $u_3\geq0$, $u_1-u_2-u_3\geq0$ and
\[ 
\left[\begin{array}{cc}
	A & 0 \\
	0 & 1 \\
\end{array}\right]^n 
\left[\begin{array}{c}
x\\ 1
  \end{array}\right]
=
\left[
\begin{array}{ccc}
v_{i_1} & v_{i_2} & v_{i_3} \\
1 & 0 & 0 \\
\end{array}
\right]
\left[
\begin{array}{c}
u_1 \\ u_2 \\ u_3
\end{array}
\right]
\]
As in the reduction from $\mathit{PHP}(m, 1)$, the $(m+1)$-th component forces
$u_1=1$ and allows us to express the constraint $u_2+u_3\leq1$ with a
homogeneous inequality. The remaining possible choices of predicate 
$T$ reduce similarly. By Theorem \ref{thm: EOdecidability} in Section
\ref{sec: EO}, the Extended Orbit
Problem with target space of dimension three is in $\mathbf{PSPACE}$. 
Therefore, to solve $\mathit{PHP}(m, 2)$ in $\mathbf{PSPACE}$, it suffices
to choose nondeterministically a simple two-dimensional target $S_i$ and
proceed to solve an Extended Orbit instance.

Finally, consider the Polyhedron-Hitting Problem in the case when the
target polyhedron $P$ has dimension $m$, matching the dimension of the
ambient space $\mathbb{Q}^m$. Consider the halfspace description of~
$P$:
\[ P =\bigcap_{i=1}^s H_i = \bigcap_{i=1}^s 
\{ p\in\mathbb{Q}^m : v_i^Tp \geq c_i \} \] Define the linear
recurrence sequences $\ess_i(n)=v_i^TA^nx$.  By the Cayley-Hamilton
Theorem, the sequences $\ess_i$ satisfy a common recurrence equation
with characteristic polynomial the minimal polynomial $f_A(x)$ of
$A$. Define also the sequences $\ess'_i(n)=\ess_i(n)-c_i$.  It is not
difficult to show that the latter also satisfy a common recurrence
equation, with characteristic polynomial $(x-1)f_A(x)$.  Since $f_A$
has degree at most $m$, the order of the recurrence equation shared by
the sequences $\ess'_i(n)$ is at most $m+1$.  Moreover $A^nx \in P$
iff $\ess'_i(n) \geq 0$ for $i=1,\ldots,s$.

Thus, $\mathit{PHP}(m,m)$ reduces to the \textbf{Simultaneous Positivity
Problem}: given a family of linear recurrence sequences $\ess'_i(n)$,
$i=1,\ldots,s$, which satisfy a common recurrence relation of order $m+1$, does
there exist an index $n$ such that $\ess'_i(n)\geq0$ for all $i$?  This problem
is the focus of Section \ref{sec: simpos}, where we place it in
$\mathbf{PSPACE}$ in the case of LRS over $\ra$ whose shared recurrence
relation is of order at most three, or of order four but with $1$ as a
characteristic root. This immediately shows that $\mathit{PHP}(3,3)$ is in
$\mathbf{PSPACE}$, completing the proof of Theorem \ref{thm:
PHPdecidability}.\footnote{In fact we can solve the problem in greater
generality. One can show a $\mathbf{PSPACE}$ bound in the case
of a \emph{simple} shared recurrence with at most four dominant
complex roots. This in turn entails membership in $\mathbf{PSPACE}$ for 
$\mathit{PHP}(4,4)$ and $\mathit{PHP}(5,5)$ in the case of a diagonalisable 
matrix. We omit this from the present paper for lack of space.}

\subsection{Hardness results}\label{subsec: PHPhard}

Now we proceed to give hardness results for the Polyhedron-Hitting Problem.
First, observe that lower-dimensional versions of $\mathit{PHP}$ reduce to
higher-dimensional ones:
\begin{lemma}\label{lem: propagation}
For all $m,k$ such that $m\geq k$, $\mathit{PHP}(m,k)$ reduces to
$\mathit{PHP}(m+1,k)$ and to $\mathit{PHP}(m+1, k+1)$. 
\end{lemma}
\begin{proof}
Given $A\in\mathbb{Q}^{m\times m}$, $x\in\mathbb{Q}^m$ and a polyhedron
$P\subseteq\mathbb{Q}^m$ with $\dim(P)=k$, we define the polyhedra $P' =
\{ (t, 0) \in\mathbb{Q}^{m+1} : t\in P\}$ and $P'' = \{ (t, 1) \in
\mathbb{Q}^{m+1} : t \in P \}$. Note that $\dim(P')= k$ and $\dim(P'') = k+1$.
Then
\begin{eqnarray*}
A^nx\in P
& \iff &
\left[\begin{array}{cc}
	A & 0 \\
	0 & 1 \\
\end{array}\right]^n
\left[\begin{array}{c} x \\ 0 \end{array}\right] \in P' \\
& \iff &
\left[\begin{array}{cc}
	A & 0 \\
	0 & 1 \\
\end{array}\right]^n
\left[\begin{array}{c} x \\ 1 \end{array}\right] \in P''
\end{eqnarray*}
which shows both reductions.
\end{proof}

Next, recall that \textbf{Skolem's Problem} is the problem of determining,
given a linear recurrence sequence $\ess(n)$ over $\mathbb{Q}$, whether it has
a zero, that is, an index $n\in\mathbb{N}$ such that $\ess(n)=0$. The
decidability of Skolem's Problem for sequences of order $5$ or greater has been
open for decades.

It is easy to show that Skolem's Problem for LRS of order $m$ reduces to
$\mathit{PHP}(m, m-1)$. For a linear recurrence sequence $\ess(n)=y^TA^nx$, we
have $\ess(n)=0$ if and only if $A^nx \in P$, where $P$ is the polyhedron $\{ t
\in \mathbb{Q}^m : y^Tt \geq 0\mbox{ and }y^Tt \leq 0 \}$. In fact,
$P=(\spa\{y\})^\perp$, so $\dim(P)=m-1$ and this is an instance of
$\mathit{PHP}(m, m-1)$. By Lemma \ref{lem: propagation}, it follows that 
whenever $m>k$, decidability of $\mathit{PHP}(m,k)$ would imply 
decidability of Skolem's Problem for LRS of order $k+1$.

In fact, we can show that even $\mathit{PHP}(4,3)$ is hard for
Skolem's Problem for linear recurrence sequences of order $5$.
\begin{lemma}\label{lem: skolem5hardness}
Skolem's Problem for LRS of order $5$ reduces to $\mathit{PHP}(4,3)$.
\end{lemma}
\begin{proof}
As discussed in reference \cite{rp}, the only outstanding case of Skolem's
Problem of order $5$ is when the LRS has five characteristic roots: two pairs
of complex conjugates $\lambda_1,\overline{\lambda_1}$,
$\lambda_2,\overline{\lambda_2}$ and a real root $\rho$, such that
$|\lambda_1|=|\lambda_2|>|\rho|>0$.  Therefore, let $\ess_1(n)$ be such a
sequence, given by
\[ 
\ess_1(n) = 
a\lambda_1^n + \overline{a\lambda_1^n} + 
b\lambda_2^n + \overline{b\lambda_2^n} + c\rho^n 
\]
Define the order-4 sequence $\ess_2(n)$ by
\[ 
\ess_2(n) = 
\frac{
a\lambda_1^n + 
\overline{a\lambda_1^n} + 
b\lambda_2^n + 
\overline{b\lambda_2^n}}{\rho^n}
\]
Let $A$ be the $4\times4$ companion matrix of $\ess_2$, and let
$x$ be the vector of initial terms of $\ess_2$, so that
\[ A^n x = \left[ \begin{array}{c} \ess_2(n) \\ \ess_2(n+1) \\ \ess_2(n+2) \\ \ess_2(n+3)\end{array} \right] \]
Then $\ess_1(n)=0$ if and only if $\ess_2(n)=-c$, or equivalently, 
if there exist $u_1,u_2,u_3$ such that
\[ 
A^n x = 
\left[ \begin{array}{c} -c \\ 0 \\ 0 \\ 0 \\ \end{array} \right] 
+ u_1\left[\begin{array}{c} 0 \\ 1 \\ 0 \\ 0 \\ \end{array} \right]
+ u_2\left[\begin{array}{c} 0 \\ 0 \\ 1 \\ 0 \\ \end{array} \right]
+ u_3\left[\begin{array}{c} 0 \\ 0 \\ 0 \\ 1 \\ \end{array} \right]
\]
which is an instance of $\mathit{PHP}(4,3)$.
\end{proof}

Finally, in Section \ref{sec: simpos} we show that solving $\mathit{PHP}(m,k)$
for $m\geq k\geq4$ is highly unlikely without major breakthroughs in analytic
number theory. For any real number $x$, the \emph{homogeneous Diophantine
approximation type} $L(x)$ is a measure of the extent to which $x$ can be
well-approximated by rationals. It is defined by: 
\[ L(x) = \inf\left\{c\in\mathbb{R}:
\exists n,m\in\mathbb{Z}.\left|x-\frac{n}{m}\right|< \frac{c}{m^2}
\right\} 
\] 
Much effort has been devoted to the study of the possible values of the
approximation type, see for instance \cite{lagrangeBook}. Nonetheless, very
little is known about the approximation type of the vast majority of
transcendental numbers.  In Section \ref{sec: simpos} we prove that a decision
procedure for the Simultaneous Positivity Problem for rational recurrences
order at most $4$ would entail the computability of $L(\arg\lambda/2\pi)$ for
any complex number $\lambda\in \mathbb{Q}(i)$ of absolute value
$1$.\footnote{Recall that a real number $x$ is \emph{computable} if there
exists an algorithm which, given any rational $\varepsilon>0$ as input,
computes a rational $q$ such that $|q-x|<\varepsilon$.} Therefore, a decision
procedure for $\mathit{PHP}(4,4)$ is extremely unlikely without significant
advances in Diophantine approximation. By Lemma \ref{lem: propagation}, the
same hardness result holds for $\mathit{PHP}(m,k)$ with $m\geq k\geq4$.  A
similar result has been shown in~\cite{positivity} concerning the Positivity
Problem for single linear recurrence sequences of order at most 6.


Our results are summarised in tabular form in the figure presented in the Introduction.

\section{Extended Orbit Problem}\label{sec: EO}

In this section, we give an overview of the
\emph{Extended Orbit Problem}.  
We are given a matrix $A\in\mathbb{Q}^{m\times m}$, an
initial point $x\in\mathbb{Q}^m$, and a \emph{target cone} specified
by a set of vectors $\{ y_1, \dots, y_d \}\subseteq\mathbb{Q}^m$ and a
constraint matrix $B\in(\ra)^{k\times d}$.  The question is whether
there exists an exponent $n\in\mathbb{N}$ and coordinates
$u=(u_1,\dots,u_d)^T$ such that $A^nx = \sum_{i=1}^d u_i y_i$ and
$Bu\geq0$.  We refer to the space $V$ spanned by
$y_1,\ldots,y_d$, which contains the target cone, as the \emph{target
  space}.

Our main decidability result concerning the Extended Orbit Problem
is the following: 
\begin{theorem}\label{thm: EOdecidability} 
The Extended Orbit Problem is in $\mathbf{PTIME}$ in the case of a
one-dimensional target space, and in $\mathbf{PSPACE}$ in the case of a two- or
three-dimensional target space.  
\end{theorem} 
Notice that these complexity bounds depend only on the dimension of the target
space $V$, not on the dimension of the ambient space $\mathbb{Q}^m$. 


We now give an overview of the strategy for proving Theorem \ref{thm:
EOdecidability}, and consign the full proof 
to Appendix
\ref{app: EO3} in the interest of clarity.
The decision method constructs a `Master System' consisting of
equations in $n$ and $u=(u_1,\dots,u_d)$ together with the inequalities
$Bu\geq0$ given as part of the input. The solutions of the Master System
are in one-to-one correspondence with the solutions of the problem instance.

When the Master System contains sufficiently many equations, it was
shown in reference \cite{orbitArxiv} that a bound $N$ can be derived
such that if $n>N$, then $A^nx \not\in V$. Writing $||I||$ for the
size of the input, we have $N\in\pin{||I||}$ when $\dim(V)=1$ and
$N\in\ein{||I||}$ when $\dim(V)\leq3$. 

With a one-dimensional target,
it is sufficient to try all exponents $n\leq N$ to get a
polynomial-time algorithm. In the two- and three-dimensional case, the
algorithm is a guess-and-check procedure. An exponent $n\leq N$ is
nondeterministically chosen as a possible witness. Then $A^nx\in V$ is
verified by checking whether the determinant of the matrix with
columns $A^nx,y_1,\dots,y_d$ equals $0$. If it does, then $A^nx\in V$,
and we proceed to calculate the coefficients $u_1,\dots,u_d$
witnessing this membership and to verify the inequalities $Bu\geq0$
which they must satisfy. 

In the verification procedure, all numbers
are expressed as arithmetic circuits and exponentiation is performed
using repeated squaring. Recall that $\mathbf{PosSLP}$ is the class of
problems which reduce in polynomial time to checking whether an
arithmetic circuit evaluates to a positive number. The described
operations may all be carried out using an oracle for
$\mathbf{PosSLP}$, so the algorithm gives a complexity upper bound of
$\mathbf{NP}^\mathbf{PosSLP}$ for the case of a large Master System.
The work of Allender et al.~\cite{allender} places $\mathbf{PosSLP}$
in the counting hierarchy, which shows the algorithm runs in
polynomial space, as Theorem \ref{thm: EOdecidability} claims.

On the other hand, when the Master System contains `few' equations, we
do not have such a bound $N$ beyond which membership in $V$ is
impossible. These cases are the main focus of Appendix \ref{app: EO3},
where we show how to solve such small Master Systems. The procedure
invokes a decision method for the Simultaneous Positivity Problem,
which is discussed in Section \ref{sec: simpos} and is shown to be in
$\mathbf{PSPACE}$ for all orders which arise in the reduction.

\section{Simultaneous Positivity}\label{sec: simpos}

In this section, we consider the \textbf{Simultaneous Positivity problem}:
given linear recurrence sequences $\ess_1,\dots,\ess_k$ over $\ra$ which
satisfy a common recurrence equation, are they ever simultaneously positive,
that is, does there exist $n$ such that $\ess_i(n) \geq 0$ for all
$i\in\{1,\dots,k\}$? Solving this problem is instrumental in our decision
procedures for both the Extended Orbit Problem and the Polyhedron-Hitting
Problem.

The asymptotic behaviour of a linear recurrence sequence $\ess$ is
closely linked to its dominant characteristic roots, that is, the
characteristic roots of greatest magnitude. If
$\lambda_1,\dots,\lambda_s$ are the dominant roots, we can write \[
\frac{\ess(n)}{|\lambda_1|^n} =
P_1(n)\left(\frac{\lambda_1}{|\lambda_1|}\right)^n + \dots +
P_s(n)\left(\frac{\lambda_s}{|\lambda_1|}\right)^n + r(n) \] where
$r(n)$ tends to $0$ exponentially quickly. We can use the polynomial
root-separation bound (\ref{eq:root separation}) in Appendix \ref{sec: 
algnum} to bound the absolute
value of the quotient $\lambda/\lambda_1$, where $\lambda$ is a
non-dominant characteristic root.  Thus we can show:
\begin{lemma}\label{lem: lrsShrinking}
Suppose we are given an LRS $\ess$ as above. Then there exist constants
$\varepsilon\in(0,1)$ and $N\in\mathbb{N}$ such that $N\in\ein{||\ess||}$,
$\varepsilon^{-1}\in\ein{||\ess||}$, and $|r(n)| < (1-\varepsilon)^n$
for all $n>N$.
\end{lemma}

\subsection{Decidability results}

In this section we prove the following result:
\begin{theorem}\label{thm: simposDecidability} 
The Simultaneous Positivity Problem is in $\mathbf{PSPACE}$ for sequences over
$\ra$ whose common recurrence equation has order at most $3$, or order $4$ but
with at least one real root.  
\end{theorem}
We will restrict our attention to non-degenerate LRS. As outlined in
Appendix \ref{subsec: LRSprelims}, a degenerate sequence can be
partitioned into non-degenerate subsequences. Then the Simultaneous
Positivity instance is equivalent to the disjunction of all
Simultaneous Positivity instances where each degenerate sequence has
been replaced by one of its non-degenerate subsequences.  In general,
this leads to exponentially many non-degenerate problem instances.
However, this leaves Theorem \ref{thm: simposDecidability} unaffected,
as a non-degenerate problem instance may simply be guessed
nondeterministically by a $\mathbf{PSPACE}$ algorithm.

The assumption of non-degeneracy guarantees that there can be at most one real
root among the dominant roots of the sequences. We can assume without loss of
generality that any real root of the sequence is positive (otherwise we
separately consider the cases of even and odd $n$).

The algorithm for Simultaneous Positivity is similar to the one for
Extended Orbit. We search for witnesses up to some computable bound
$N\in\ein{\Vert I\Vert}$.  To this end, we will choose a witness $n$
nondeterministically and then verify $\ess_j(n)=v_j^TM_j^nw_j\geq0$
for all $j$.  Recalling that the entries of $M_j$ are algebraic
numbers, we can verify this family of inequalities by constructing a
sentence $\tau$ in the existential first-order theory of the reals
which is true if and only if $v_j^TM_j^nw_j\geq0$ for all $j$. We
specify each real algebraic number with description $(f_\alpha, x_0,
y_0, R)$ using the first-order formula $\exists z. f_\alpha(z) = 0
\wedge (z-x_0)^2 + y_0^2 \leq R^2$. To ensure that
$\Vert\tau\Vert\in\pin{\Vert I\Vert}$, we use repeated squaring to
calculate $M_j^n$. Finally, we check the validity of $\tau$ in
$\mathbf{PSPACE}$, as per Theorem \ref{thm: fo} in Appendix \ref{sec: 
FOreals}.

We now consider two cases, according to the number of dominant complex
roots of the shared recurrence equation.

\textbf{No dominant complex roots.} Suppose the dominant
characteristic roots do not include a pair of complex conjugates. Then
by the assumption of non-degeneracy, there is one real dominant root
$\rho>0$.  Then the $j$-th sequence is given by \[
\frac{\ess_j(n)}{\rho^n} = P_j(n) + r_j(n) \] where $r_j$ is itself a
linear recurrence of lower order which converges to $0$ exponentially
quickly, and $P_j\in(\ra)[x]$.  Each polynomial $P_j(n)$ is either
identically zero or is ultimately positive or ultimately negative as
$n$ tends to infinity.   In the
latter two cases, there is an effective threshold
$N_j\in\ein{||\ess_j(n)||}$ beyond which the sign of $\ess_j$ does not
change.  If some $\ess_j$ is ultimately negative, then any witness to
the problem instance must be bounded above by $N_j$. Since $N_j$ is at
most exponentially large in the size of the input, we use a
guess-and-check procedure and are done.  Similarly, for each sequence
$\ess_j$ for which $P_j$ is ultimately positive we can search for
witnesses up to the threshold $N_j$ and if none are found, we discard
$\ess_j$ as if it were uniformly positive. Finally, we are left only
with sequences $\ess_j$ for which $P_j$ is identically zero. Then the
problem instance is equivalent to Simultaneous Positivity on the
sequences $r_j$. These sequences satisfy a common recurrence equation
of lower order, so we proceed inductively.

\textbf{Two simple dominant complex roots.}
Suppose now that the dominant roots of the shared recurrence equation
include exactly two complex roots $\lambda,\overline{\lambda}$ and possibly
a real dominant root $\rho_1>0$. Moreover,
assume that the roots are all simple, so the $j$-th sequence 
is given by
\[
\ess_j(n) = a_j\lambda^n + \overline{a_j}\overline{\lambda}^n + 
b_j\rho_1^n + c_j\rho_2^n
\]
that is,
\[ \frac{\ess_j(n)}{|\lambda|^n}
= 2|a_j|\cos(\alpha_j+n\varphi) + b_j + r_j(n) \] where $\alpha_j =
\arg(a_j)$ and $\varphi = \arg(\lambda)$. Moreover, $r_j$ is a linear
recurrence sequence of order at most $2$ with real characteristic
roots.  Observe that for all $j$, $b_j+r_j(n)$ is either ultimately
positive or ultimately negative as $n$ tends to infinity. Furthermore,
a threshold beyond which the sign does not change is effectively
computable and at most exponential in $||\ess_j||$. Following the
reasoning of the previous case, we see that we can dismiss sequences
$\ess_j$ which have $a_j=0$.

Assume therefore that $a_j\neq 0$ for all $j$. By Lemma \ref{lem:cosExpShrink}
in Appendix \ref{app: baker},
for each sequence $\ess_j$ there exists an effective
threshold $N_j\in\ein{||\ess_j||}$ such that for $n>N_j$, $r_j(n)$ is too
small to influence the sign of $\ess_j(n)$. That is, for all $n>N_j$, we have 
\[  \ess_j(n) \geq 0 \iff b_j+\cos(\alpha_j+n\varphi) \geq 0  \]
Therefore, for $n>N=\max_j\{N_j\}$, the problem instance is equivalent 
to a conjunction of inequalities in $n$:
\[ \forall j . \cos(\alpha_j + n\varphi) \geq -b_j \]
We use guess-and-check to look for witnesses $n\leq N$. If none are found,
the problem instance is then decidable in $\mathbf{PSPACE}$ 
by Lemma \ref{lem: arcInequalities} in Appendix \ref{app: baker}.

\subsection{Hardness}

We now proceed to show our main hardness result for Simultaneous Positivity and
hence for $\mathit{PHP}(m, m)$. Recall that the homogeneous Diophantine
approximation type $L(x)$ of a real number $x$, 
defined in Section \ref{subsec: PHPhard}, is a
measure of how well $x$ can be approximated by rationals. Very little progress
has been made on calculating the approximation type for the vast majority of
transcendental numbers. In this section, we show that a decision procedure for
Simultaneous Positivity for LRS with shared recurrence equation of order $4$ 
would entail the computability of the approximation type of all Gaussian 
rationals:

\begin{theorem}\label{thm: simposHardness}
Suppose that Simultaneous Positivity is decidable for rational linear
recurrence sequences. Then for any $\lambda\in \mathbb{Q}(i)$
on the unit circle, $L(\arg\lambda/2\pi)$ is a computable number.
\end{theorem}

Suppose we wish to calculate $L(\varphi/2\pi)$, where $\varphi=\arg
\lambda$ for some $\lambda$ of magnitude $1$.  Consider the
following two sequences for some fixed rational number $A$:
\begin{align*}
\ess_1(n) = & \frac{1}{2}\left((A-in)\lambda^n + (A+in) 
\overline{\lambda}^n\right)\\[2pt]
\ess_2(n) = & \frac{1}{2}\left((A+in)\lambda^n + (A-in) 
\overline{\lambda}^n\right)
\end{align*}

It is straightforward to verify that $\ess_1(n)$ and $\ess_2(n)$ are
both rational sequences satisfying a common order-4 recurrence with
characteristic polynomial $(x-\lambda)^2(x-\overline{\lambda})^2$.
Moreover we have
\begin{eqnarray*}
\ess_1(n) & = & n\cos(n\varphi - \pi/2) + A\cos(n\varphi) \\
 & = &  A\cos(n\varphi) + n\sin(n\varphi)
\end{eqnarray*}
\begin{eqnarray*}
\ess_2(n) & = & n\cos(n\varphi + \pi/2) + A\cos(n\varphi) \\
 & = &  A\cos(n\varphi) - n\sin(n\varphi) 
\end{eqnarray*}
Let $w_n = n|\sin(n\varphi)| - A\cos(n\varphi)$. It is clear
that $\ess_1(n)\geq0$ and $\ess_2(n)\geq0$ if and only if $w_n \leq 0$.
We will
show that a Simultaneous Positivity oracle may be used on these sequences
for different choices of $A$ to compute arbitrarily good approximations
of $L(\varphi/2\pi)$. Throughout this section, write $[x]$ to denote the
distance from $x$ to the closest integer multiple of $2\pi$, that is,
$[x] = \min\{|x-2\pi j| : j\in\mathbb{Z}\}$.

Given $\varepsilon\in(0, 1)$, there exists $\delta>0$ such that
for all $x\in[-\delta, \delta]$, the following hold:
\begin{equation}\label{eq:sineIneq}
(1-\varepsilon)|x| \leq |\sin x| \leq |x|
\end{equation}
\begin{equation}\label{eq:cosIneq}
1-\varepsilon \leq \cos x
\end{equation}
Moreover, there exists $N\in\mathbb{N}$ such that $A/N\leq\delta$ and also,
\begin{equation}\label{eq:withinDelta}
\mbox{ if $|\sin x| \leq A/N$, then $|x|\leq\delta$. }
\end{equation}

\begin{lemma}\label{lem:hardness1}
Suppose that $n\geq N$ is such that $w_n \leq 0$.  Then
$n[n\varphi]<A/(1-\varepsilon)$. 
\end{lemma}

\begin{proof}
\begin{align*}
& |\sin(n\varphi)| \leq \frac{A}{n}\cos(n\varphi) \leq \frac{A}{N} & \mbox{[as $w_n\leq0, n\geq N$]} \\
\\
\Rightarrow & [n\varphi] \leq \delta & \mbox{ [by (\ref{eq:withinDelta})]}\\
\end{align*}
But from the definition of $w_n$, inequality (\ref{eq:sineIneq}) and $\cos x\leq 1$,
we have
\[
w_n = n|\sin(n\varphi)| - A\cos(n\varphi)
\geq n(1-\varepsilon)[n\varphi] - A 
\]
Therefore, $n[n\varphi]\leq A/(1-\varepsilon)$.
\end{proof}

\begin{lemma}\label{lem:hardness2}
Let $n\geq N$ be such that $n[n\varphi] \leq A(1-\varepsilon)$.  Then
$w_n \leq 0$.
\end{lemma}

\begin{proof} Notice that
\[ [n\varphi] \leq \frac{A(1-\varepsilon)}{n} \leq \frac{A}{N} \leq \delta \]
so for $w_n$ we have
\begin{align*}
w_n = & n|\sin(n\varphi)| - A\cos(n\varphi) & \mbox{[definition of $w_n$] } \\
\\
    \leq & n[n\varphi] - A(1-\varepsilon) & \mbox{[by (\ref{eq:sineIneq})(\ref{eq:cosIneq})] } \\
\\
    \leq & A(1-\varepsilon) - A(1-\varepsilon) = 0 & \mbox{[by premise] } \\
\end{align*}

\end{proof}

Letting $t=\varphi/2\pi$, we see that
\[
2\pi L(t) = \inf_{m\in\mathbb{N}}m[m\varphi]
\]
Thus to show computability of $L(t)$ it is enough to show that
$\inf_{m\in\mathbb{N}}m[m\varphi]$ is computable.  For this in turn it
suffices to provide a procedure that, given $a,b \in \mathbb{Q}$ with
$a<b$, computes a threshold $N \in \mathbb{N}$ and either outputs that
$\inf_{m\geq N}m[m\varphi] < b$ or $\inf_{m\geq N}m[m\varphi] >
a$.  (Clearly $\inf_{m< N}m[m\varphi]$ can be computed to any desired
precision.)

Given $a<b$ as above, compute $\varepsilon$ and $A$ such that
\[ a < A(1-\varepsilon) < \frac{A}{1-\varepsilon} < b \, .\]
Calculate also the constant $N$ in the statement of
Lemmas~\ref{lem:hardness1} and \ref{lem:hardness2} for this choice of
$\varepsilon$ and $A$.  Then run a Simultaneous Positivity oracle on
the $N$-th tails of the two sequences $\ess_1(n)$ and $\ess_2(n)$ to
determine whether $w_n\leq0$ for some $n\geq N$.  If the oracle accepts,
then $\inf_{m\in\mathbb{N}}m[m\varphi] \leq \frac{A}{1-\varepsilon} < b$
by Lemma~\ref{lem:hardness1}.  If the oracle rejects, then
$\inf_{m\in\mathbb{N}}m[m\varphi] \geq A(1-\varepsilon) > a$ by
Lemma~\ref{lem:hardness2}.

\begin{appendix}
\section{Preliminaries}

\subsection{Polyhedra and their representations}\label{sec: polyhedra}

Here we state some basic properties of polyhedra. For more
details we refer the reader to, for example
\cite{convex1,convex2,convex3}.
A \emph{halfspace} in $\mathbb{R}^d$ is the set of points
$x\in\mathbb{R}^d$ satisfying $v^Tx\geq c$ for some fixed vector
$v\in\mathbb{R}^d$ and real number $c$. A \emph{polyhedron} in
$\mathbb{R}^d$ is the intersection of finitely many halfspaces:
\begin{equation}\label{eq: halfspace rep}
P = \left\{ x\in\mathbb{R}^d :
\begin{array}{ccc}
v_1^Tx & \geq & c_1 \\
 &  & \vdots \\
v_m^Tx & \geq & c_m \\
\end{array}
\right\} 
\end{equation}
We call the set $\{(v_1,c_1),\dots,(v_m,c_m)\}$ a \emph{halfspace
  description} of a polyhedron, or simply an \emph{H-polyhedron}.  The
problem of determining a minimal subset of the inequalities (\ref{eq:
  halfspace rep}) that define the same polyhedron is called the
\emph{H-redundancy removal problem} and is solvable in polynomial time
by reduction to linear programming.  Thus, we may freely assume that
there are no redundant constraints in the descriptions of H-polyhedra.

The \emph{dimension} of a polyhedron $P$, denoted $\dim(P)$, is the dimension
of the subspace of $\mathbb{R}^d$ spanned by $P$.  The task of calculating the
dimension of an H-polyhedron, called the \emph{H-dimension problem}, can be
done in polynomial time by solving polynomially many linear programs.  If
$\dim(P) = d$, we call $P$ \emph{full-dimensional}. The minimal halfspace
representation of a full-dimensional polyhedron is unique, up to scaling of the
inequalities in (\ref{eq: halfspace rep}).

The \emph{convex cone} of a finite set of vectors $v_1,\dots,v_m$ is defined as
\[
\mbox{cone}(\{v_1,\dots,v_m\}) =
\{ \lambda_1v_1 + \dots + \lambda_mv_m : 
\forall i.\lambda_i\geq 0 \}
\]
If the vectors $v_1,\dots,v_m$ are linearly independent, the cone is called
\emph{simplicial}. A classical result, due to Carath{\'e}odory, states that
each finitely generated cone can be written as a finite union of simplicial
cones:
\begin{theorem}\label{thm: cara}
(Carath{\'e}odory)
Let $v_1,\dots,v_m\in\mathbb{R}^d$. If $v\in\mbox{cone}(v_1,\dots,v_m)$, then
$v$ belongs to the cone generated by a linearly independent subset of $\{v_1,
\dots,v_m\}$.
\end{theorem}
We use this to prove that any two-dimensional polyhedron decomposes
into a finite union of simple two-dimensional shapes:
\begin{lemma}\label{lem: poly2}
Suppose $P\subseteq\mathbb{R}^d$ is a two-dimensional polyhedron. Then
$P = \bigcup_{i=1}^m A_i$, where $m$ is finite and each of $A_i$ is of
the form
\[ A_i =
\{ u_i + \alpha v_i + \beta w_i : T_i(\alpha, \beta) \} \]
for vectors $u_i,v_i,w_i\in\mathbb{R}^d$ and predicates $T_i(\alpha,\beta)$
chosen from the following:
\begin{itemize}
\item $T_i(\alpha,\beta)\equiv\alpha\geq0\wedge\beta\geq0$ ($A_i$ is
an infinite cone)
\item $T_i(\alpha,\beta)\equiv\alpha\geq0\wedge\beta\geq0\wedge\alpha+\beta
\leq1$ ($A_i$ is a triangle)
\item $T_i(\alpha,\beta)\equiv\alpha\geq0\wedge\beta\geq0\wedge\beta\leq1$
($A_i$ is an infinite strip)
\end{itemize}
Furthermore, if we are given a halfspace description of $P$ with 
length $\Vert P\Vert$, the size of the representation of each vector $u_i,v_i,w_i$ 
is at most $\pin{\Vert P\Vert}$.
\end{lemma}

\begin{proof}
Let 
\[P = \{ x\in\mathbb{R}^d : Ax\geq b\}\]
for some $A\in\mathbb{R}^{m\times d}$, $b\in\mathbb{R}^d$ and define
the polygon 
\[ P' = 
\{ y\in\mathbb{R}^{d+1} : 
[\begin{array}{cc}A & -b\end{array}]\,y \geq 0 \}
\]
so that $\dim(P')=3$ and
\[
P=\{ x\in\mathbb{R}^d : (\begin{array}{cc} x & 1 \end{array})^T\in P' \}
\]
Notice that $P'$ is specified using only homogeneous inequalities, so 
there exist vectors $V=\{v_1,\dots,v_s\}$ such that $P'=\mbox{cone}(V)$.
By scaling if necessary, we can assume the $(d+1)$-th component of each $v_i$
is either $0$ or $1$. Let $\mathcal{H}$ denote the hyperplane in 
$\mathbb{R}^{d+1}$ where the $(d+1)$-th coordinate is $1$. 
By Carath{\'e}odory's Theorem, $P'$ may be written as the union of finitely 
many cones generated from linearly independent subsets of $V$. Let $u_i$ 
be the projection of $v_i$ to the first $d$ coordinates. Since
$\dim(P')=3$, no more than three elements of $V$ can be linearly independent,
so
\[
P' = \bigcup_{(i_1,i_2,i_3)\in I} \mbox{cone}(v_{i_1},v_{i_2},v_{i_3})
\]
The intersection $\mathcal{H}\cap\mbox{cone}(v_{i_1},v_{i_2},v_{i_3})$ 
is non-empty if and only if at least one of $v_{i_1},v_{i_2},v_{i_3}$ has $1$
in the $(d+1)$-th coordinate. Therefore,
$P$ is the finite union of shapes $A_i$ with only two
degrees of freedom:
\[
A_i = \{ \alpha u_{i_1} + \beta u_{i_2} + \gamma u_{i_3} : 
\alpha,\beta,\gamma\geq0\wedge T_i(\alpha,\beta,\gamma) \}
\]
where each predicate $T_i$ is $\alpha=1$, or $\alpha+\beta=1$, or 
$\alpha+\beta+\gamma=1$. These are precisely the desired three types of 
parametric shapes. The descriptions of the vectors involved is polynomially
large because each vector $v_i$ is the intersection of $d$ of the 
halfspaces in $\mathbb{R}^{d+1}$ which define $P'$.
\end{proof}
A simpler version of the above result gives a similar parametric
form in the case $\dim(P)=1$:
\begin{lemma}\label{lem: poly1}
Suppose $P\subseteq\mathbb{R}^d$ is a one-dimensional polyhedron. Then 
\[ P = \{ v_1 + \alpha v_2 : T(\alpha)\} \]
where the predicate $T(\alpha)$ is one of $\alpha\in\mathbb{R}$, $\alpha\geq0$
and $\alpha\in[0,1]$.
Furthermore, if we are given a halfspace description of $P$ with 
length $\Vert P\Vert$, the size of the representation of $v_1,v_2$ 
is at most $\pin{\Vert P\Vert}$.
\end{lemma}

\subsection{Algebraic numbers}\label{sec: algnum}
In this section we briefly review relevant notions in algebraic number
theory. See, e.g.,~\cite{cohen} for more details.

A complex number $\alpha$ is \emph{algebraic} if there exists a
polynomial $p\in\mathbb{Q}[x]$ such that $p(\alpha)=0$.  The set of
algebraic numbers, denoted by $\mathbb{A}$, is a subfield of
$\mathbb{C}$. The \emph{minimal polynomial} of $\alpha$, denoted
$f_\alpha(x)$, is the unique monic polynomial with rational
coefficients of least degree which vanishes at $\alpha$.  The
\emph{degree} of $\alpha\in\mathbb{A}$ is defined as the degree of its
minimal polynomial and is denoted by $n_{\alpha}$. The \emph{height}
of $\alpha$ is defined as the maximum absolute value of a numerator or
denominator of a coefficient of the minimal polynomial of $\alpha$,
and is denoted by $H_{\alpha}$. The roots of $f_{\alpha}(x)$
(including $\alpha$) are called the \emph{Galois conjugates} of
$\alpha$. An \emph{algebraic integer} is an algebraic number $\alpha$
such that $f_{\alpha}\in\mathbb{Z}[x]$. The set of algebraic integers,
denoted $\mathcal{O}_\mathbb{A}$, is a ring under the usual addition
and multiplication.

The \emph{canonical representation} of an algebraic number $\alpha$ is its
minimal polynomial $f_{\alpha}(x)$, along with a numerical approximation of
$\mathit{Re}(\alpha)$ and $\mathit{Im}(\alpha)$ of sufficient precision to
distinguish $\alpha$ from its Galois conjugates. More precisely, we represent
$\alpha$ by the tuple 
\[
(f_{\alpha},x,y,R)
\in\mathbb{Q}[x]\times\mathbb{Q}^3
\]
meaning that $\alpha$ is the unique root of $f_{\alpha}$ inside the circle
centred at $(x,y)$ in the complex plane with radius $R$. A bound due to
Mignotte \cite{mignotteRootSep} states that for roots $\alpha_i\neq\alpha_j$ of
a polynomial $p(x)$,
\begin{equation}\label{eq:root separation}
|\alpha_i-\alpha_j|>\frac{\sqrt{6}}{n^{(n+1)/2}H^{n-1}}
\end{equation}
where $n$ and $H$ are the degree and height of $p$, respectively.  Thus, if $R$
is restricted to be less than a quarter of the root separation bound, the
representation is well-defined and allows for equality checking. Observe that
given $f_{\alpha}$, the remaining data necessary to describe $\alpha$ is
polynomial in the length of the input. It is known how to obtain polynomially
many bits of the roots of any $p\in\mathbb{Q}[x]$ in polynomial time
\cite{panApproximatingRoots}.

When we say an algebraic number $\alpha$ is given, we
assume we have a canonical description of $\alpha$. We will denote by
$\Vert\alpha\Vert $ the length of this description, assuming that integers are
expressed in binary and rationals are expressed as pairs of integers. Observe
that $|\alpha|$ is an exponentially large quantity in $\Vert\alpha\Vert$
whereas $\ln|\alpha|$ is polynomially large. Notice also that $1/\ln|\alpha|$
is at most exponentially large in $\Vert\alpha\Vert$.  For a rational $a$,
$\Vert a\Vert$ is just the sum of the lengths of its numerator and denominator
written in binary. For a polynomial $p\in\mathbb{Q}[x]$, $\Vert p\Vert$ will
denote $\sum_{i=0}^n\Vert p_i\Vert$ where $n$ is the degree of the polynomial
and $p_i$ are its coefficients. Using the resultant method, operations may be
performed efficiently on algebraic numbers. Specifically, techniques from 
algebraic number theory \cite{cohen} yield the following lemma:
\begin{lemma}\label{lem:operations on algebraic numbers}
Given canonical representations of $\alpha,\beta\in\mathbb{A}$ and a
polynomial $p\in\mathbb{Q}[x]$, it is possible to compute canonical
descriptions of $\alpha\pm\beta$, $\alpha\beta^{\pm1}$,
$\sqrt{\alpha}$ and $p(\alpha)$, to check the equality $\alpha=\beta$
and $\alpha$'s membership in $\mathbb{N}, \mathbb{Z}, \mathbb{Q}$, and
finally to determine whether $\alpha$ is a root of unity, and if so,
to calculate its order and argument.  All of these procedures have
polynomial running time.
\end{lemma}

\subsection{Linear recurrence sequences}\label{subsec: LRSprelims}

We now recall some basic properties of linear recurrence sequences.
For more details, we refer the reader to \cite{everest,tucs}. 
A \emph{real linear recurrence sequence (LRS)} is an infinite sequence
$\ess=\langle \ess(0),\ess(1),\ess(2),\dots\rangle$ over $\mathbb{R}$
such that there exists a natural number $k$ and real numbers $a_1,\dots,a_k$
such that $a_k\neq0$ and $\ess$ satisfies the linear recurrence equation
\begin{equation}\label{eq: recurrence}
\ess(n+k) = a_1\ess(n+k-1) + a_2\ess(n+k-2) + \dots + a_k\ess(n)
\end{equation}
The recurrence (\ref{eq: recurrence}) is said to have order $k$.  Note
that the same sequence can satisfy different recurrence relations, but
it satisfies a unique recurrence of minimum order.

The \emph{characteristic polynomial} of $\ess$ is 
\[ p(x) = x^k - a_1x^{k-1} - a_2x^{k-2} - \dots - a_k \]
and its roots are called the \emph{characteristic roots} of the sequence. For
real LRS, the set of characteristic roots is closed under complex conjugation.
If $\rho_1,\dots,\rho_l\in\mathbb{R}$ are the real roots of $p(x)$ and
$\gamma_1,\overline\gamma_1,\dots,\gamma_m,\overline\gamma_m\in\mathbb{C}$ are 
the complex ones, the sequence is given by
\[ 
\ess(n) = \sum_{i=1}^l A_i(n)\rho_i^n + 
\sum_{j=1}^m\left( C_j(n)\gamma_j^n + 
\overline{C_j}(n)\overline{\gamma_j}^n \right)
\]
for all $n\geq0$,
where $A_i\in\mathbb{R}[x]$ and $C_j\in\mathbb{C}[x]$ are
univariate polynomials whose degrees are at most the multiplicity
of the corresponding roots of $p(x)$. The coefficients of $A_i,C_i$
are effectively computable algebraic numbers.

If $M\in\mathbb{R}^{k\times k}$ is a real square matrix and $v,w\in
\mathbb{R}^k$ are real column vectors, then it can be shown using the
Cayley-Hamilton Theorem that the sequence $\ess(n)=v^TM^nw$ satisfies
a linear recurrence of order $k$.  Conversely, any LRS may be
expressed in this way: it is sufficient to take $M$ to be the
transposed companion matrix of the characteristic polynomial of
$\ess$, $v$ to be the vector $(\ess(k-1),\dots,\ess(0))^T$ of initial
terms of $\ess$ in reverse order, and $w$ to be the unit vector
$(0,\dots,0,1)^T$.  The characteristic roots of the LRS are precisely
the eigenvalues of $M$.


A linear recurrence sequence is called \emph{degenerate} if for some
pair of distinct characteristic roots $\lambda_1, \lambda_2$ of its
minimum-order recurrence, the ratio $\lambda_1/\lambda_2$ is a root of
unity, otherwise the sequence is \emph{non-degenerate}. As pointed out
in \cite{everest}, the study of arbitrary LRS can effectively be
reduced to that of non-degenerate LRS by partitioning the original LRS
into finitely many non-degenerate subsequences. Specifically, for a
given degenerate linear recurrence sequence $\ess$ with characteristic
roots $\lambda_i$, let $L$ be the least common multiple of the orders
of all ratios $\lambda_i/\lambda_j$ which are roots of unity. Then
consider the sequences \[\ess^{(j)}(n) = u^T A^{nL+j} v = u^T (A^L)^n
(A^jv)\] where $j\in\{0,\dots,L-1\}$. Each of these sequences has
characteristic roots $\lambda_i^L$ and is therefore non-degenerate,
because $(\lambda_1/\lambda_2)^{Lk} = 1$ implies $\lambda_1^L =
\lambda_2^L$. From the crude lower bound $\varphi(r)\geq\sqrt{r/2}$
on Euler's totient function, it follows that if $\alpha$ has degree $d$ and is a primitive
$r$-th root of unity, then $r\leq2d^2$. Thus, $L\in\ein{||A||}$,
so non-degeneracy can be ensured by considering at most exponentially
many subsequences of the original LRS.

\subsection{First-order theory of the reals}\label{sec: FOreals}
Let $x_1,\dots,x_m$ be first-order
variables ranging over $\mathbb{R}$, and suppose $\sigma(x_1,\dots,x_m)$ is a
Boolean combination of predicates of the form $g(x_1,\dots,x_m)\sim0$, where
$g\in\mathbb{Z}[x_1,\dots,x_m]$ is a polynomial and $\sim$ is $>$ or $=$. A 
\emph{sentence of the first-order theory of the reals} is a formula $\tau$ of
the form
\[ Q_1x_1\dots Q_mx_m \sigma(x_1,\dots,x_m) \]
where each $Q_i$ is one of the quantifiers $\exists$ and $\forall$. If all the
quantifiers are $\exists$, then $\tau$ is said to be a sentence of the 
\emph{existential} first-order theory of the reals. 

The decidability of the first-order theory of the reals was originally 
established by Tarski \cite{tarski}. Many refinements followed over the years,
culminating in the analysis of Renegar \cite{renegar}. We make use of the 
following result:

\begin{theorem}\label{thm: fo}
Suppose we are given a sentence $\tau$ of the form above using only
existential quantifiers. The problem of deciding whether $\tau$ holds
over the reals is in $\mathbf{PSPACE}$.  Furthermore, if
$M\in\mathbb{N}$ is a fixed constant and we restrict the problem to
formulae $\tau$ where the number of variables is bounded above by $M$,
then the problem is in $\mathbf{PTIME}$.
\end{theorem}

\section{Technical lemmas}\label{app: baker}
\begin{theorem}\label{thm:baker}
(Baker and W\"ustholz \cite{baker})
Let $\alpha_1,\dots,\alpha_m$ be algebraic numbers other than
$0$ or $1$, and let $b_1,\dots,b_m$ be rational integers. Write
\[ \Lambda = b_1\log\alpha_1+\dots+b_m\log\alpha_m \]
Let $A_1,\dots,A_m,B\geq e$ be real numbers such that, for
each $j\in\{1,\dots,m\}$, $A_j$ is an upper bound for the height
of $\alpha_j$, and $B$ is an upper bound for $|b_j|$. Let $d$
be the degree of the extension field $\mathbb{Q}(\alpha_1,
\dots,\alpha_m)$ over $\mathbb{Q}$. If $\Lambda\neq0$, then
\[ \log|\Lambda|>-(16md)^{2(m+2)}\log(A_1)\dots\log(A_m)\log(B)\]
\end{theorem}

\begin{theorem}\label{thm: skolem}
Suppose $\alpha,\beta,\gamma,A,B,C\in\mathbb{A}$ and the ratios of
$\alpha,\beta,\gamma$ (where they exist) are not roots of unity.  Let $\Vert
I\Vert = \Vert\alpha\Vert+\Vert\beta\Vert+\Vert\gamma\Vert+\Vert A\Vert+ \Vert
B\Vert+\Vert C\Vert$. Then there exist effective bounds $N_1\in\pin{\Vert
I\Vert}$ and $N_2\in\ein{\Vert I\Vert}$ such that
if $A\alpha^n + B\beta^n = 0$ then $n \leq N_1$,
and if $A\alpha^n + B\beta^n + C\gamma^n = 0$
or
$A\alpha^n + Bn\beta^{n-1} + C\beta^n = 0$
then $n\leq N_2$.
\end{theorem}

\begin{lemma} \label{lem:cosExpShrink}
Let $a,\lambda\in\mathbb{A}$ and $C,\chi\in\mathbb{A}\cap\mathbb{R}$ be given
where  $\lambda$ is not a root of unity and $|\chi|<|\lambda|=1$. Let
$\alpha=\arg(a)$ and $\varphi=\arg(\lambda)$.  Then there exists an effectively
computable
$N\in\mathbb{N}$ such that for all $n>N$, $|C+\cos(\alpha+n\varphi)|>|\chi|^n$.
Moreover, $N\in\ein{||I||}$ where $||I|| = ||\lambda||+||\chi||+||a||+||C||$.
\end{lemma}

\begin{proof}
Suppose that $|C| \leq 1$ and let $b=C+i\sqrt{1-C^2}=e^{i\beta}$, so that
$C=\cos(\beta)$.  Then $b$ is algebraic with $\deg(b)\in\pin{||I||}$,
$H_b\in\ein{||I||}$. It is clear that
\[ C+\cos(\alpha+n\varphi) = 
2
\cos\frac{\alpha+\beta+n\varphi}{2}
\cos\frac{\alpha-\beta+n\varphi}{2}
\]
Since $\lambda$ is not a root of unity, by Lemma \ref{thm: skolem}, there
exists an effective constant $N_1\in\pin{||I||}$ such that if $ab^{\pm
1}\lambda^n = -1$ then $n\leq N_1$. Therefore, for $n>N_1$, we have
$\cos(\alpha\pm\beta +n\varphi) \neq0$.  Let $k_n$ be the unique integer such
that $k_n\pi+(\alpha+\beta+n\varphi+\pi)/2\in[-\pi/2, \pi/2)$. Notice that
$|k_n|<2n$. Then
\begin{eqnarray*}
\left|\cos\frac{\alpha+\beta+n\varphi}{2}\right| & = &
\left|\sin\frac{\alpha+\beta+n\varphi+(2k_n+1)\pi}{2}\right| \\
& \geq &
\frac{\left|\alpha+\beta+n\varphi+(2k_n+1)\pi\right|}{2\pi} 
\end{eqnarray*}
by the inequality $|\sin(x)|\geq|x|/\pi$ for $x\in[-\pi/2,\pi/2]$. Note
that $\alpha$, $\beta$, $\varphi$ and $\pi$ are logarithms of algebraic numbers
with degree polynomial in $||I||$ and height exponential in $||I||$.
Then by from Baker's Theorem, there exist effective positive constants 
$p_1,p_2\in\pin{||I||}$ such that
\[ 
n>N_1 \Rightarrow \left|\cos\frac{\alpha+\beta+n\varphi}{2}\right| 
> (p_1n)^{-p_2} 
\]
By the same argument with $\beta$ replaced by $-\beta$, there exist effective
positive constants $N_2, p_3, p_4\in \pin{||I||}$ such that
\[ 
n>N_2 \Rightarrow \left|\cos\frac{\alpha-\beta+n\varphi}{2}\right| 
> (p_2n)^{-p_4} 
\]
However, since $\chi^n$ shrinks exponentially with $n$ and $|\chi^{-1}|
\in\ein{||I||}$, it follows that there exists an effective constant
$N_3\in\ein{||I||}$ such that for all $n>N_3$,
\[ (p_1n)^{-p_2}(p_3n)^{-p_4} > |\chi^n| \]
Then for all $n>\max\{N_1, N_2, N_3\}$, we have
\[ \left|C+\cos(\alpha+n\varphi)\right| > p_1p_3n^{-(p_2+p_4)} > |\chi^n| \]
as desired. 

The remaining case $|C|>1$ is easy. If $C>1$, we have
\[ 
C + \cos(\alpha + n\varphi) > 1 + \cos(\alpha+n\varphi) = 
\cos(0) + \cos(\alpha+n\varphi) 
\]
and the lemma follows by the above argument with $\beta = 0$. Similarly
when $C < - 1$.
\end{proof}

\begin{lemma}\label{lem: arcInequalities}
Suppose $a_1,\dots,a_m$ and $\lambda$ are all algebraic numbers on the 
unit circle and $\lambda$ is not a root of unity. Suppose also
$c_1,\dots,c_m\in\ra$. Let $\alpha_j=\arg(a_j)$ and
$\varphi=\arg(\lambda)$. Then it is decidable 
whether there exists an integer $n$ such that
\[ \cos(\alpha_j+n\varphi) \geq c_j \mbox{ for all $j=1,\dots,m$} \]
Moreover, the decision procedure's running time is $\pin{||I||}$ where
\[ ||I||=\sum_{j=1}^m\left(||a_j||+||c_j||\right)+||\lambda|| \]
\end{lemma}

\begin{proof}
Inequalities where $c_j \leq -1$ may be discarded, as they are
satisfied for all $n$, whereas the presence of inequalities with $c_j>1$
immediately makes the problem instance negative.
Now assuming $c_j\in(-1,1]$, each inequality 
\begin{equation}\label{eq:arcSituation}
	\cos(\alpha_j+n\varphi_1) \geq c_j
\end{equation}
defines an arc on the unit circle
which $\lambda^n$ must lie within. Specifically, (\ref{eq:arcSituation})
holds if and only if $\lambda^n$ lies on the arc $\mathcal{A}_j$ 
defined by
\[\mathcal{A}_j = \{ z\in\mathbb{C}:|z|=1 \mbox{ and } h(w_1,w_2,z)\leq0 \} \]
where
$ w_1=\overline{a_j}\left(c_j-i\sqrt{1-c_j^2}\right)$
and
$ w_2=\overline{a_j}\left(c_j+i\sqrt{1-c_j^2}\right)$
are the endpoints of the arc, and
\[ h(x,y,z) =
\left|\begin{array}{ccc}
	\re(x) & \im(x) & 1 \\
	\re(y) & \im(y) & 1 \\
	\re(z) & \im(z) & 1 \\
\end{array}\right|
\]
is the orientation function.\footnote{Recall that $h(x, y, z)$ is
positive if the
points $x,y,z$ (in that order) are arranged counter-clockwise 
on the complex plane, negative if they are arranged clockwise,
and zero if they are collinear.}

The endpoints of $\mathcal{A}_j$ are clearly algebraic and may be
computed explicitly in polynomial time in $||I||$. Then the intersection
$\mathcal{A}=\bigcap_j{\mathcal{A}_j}$ is also computable in polynomial time.
Since $\lambda$ is not a root of unity, the set
$\{\lambda^n:n\in\mathbb{N}\}$ is dense on the unit circle. 
If $\mathcal{A}$ is empty, then the problem instance is negative. If
$\mathcal{A}$ is a nontrivial arc on the unit circle, then by density,
the problem instance is positive. Finally, $\mathcal{A}$ could be a set of
at most two points $z_1,z_2$ on the unit circle. Then the problem instance is
positive if and only if there exists an exponent $n\in\mathbb{N}$ such that 
$\lambda^n=z_i$ for some $i$. A polynomial bound on $n$ then follows from 
Theorem \ref{thm: skolem}.
\end{proof}

\end{appendix}

\section{Extended Orbit Problem}
\label{app: EO3}

We now give the details of our decision procedure for the Extended 
Orbit Problem, as promised in Section~\ref{sec: EO}.
\subsection{A Master System}\label{subsec: EOprelims}

In \cite{orbitArxiv}, we show how to reduce the Orbit Problem
(determining whether there exists $n\in\mathbb{N}$ such that $A^nx$ lies in a
vector space $V$) to the \emph{matrix power problem}: determining whether there
exists $n\in\mathbb{N}$ such that $A^n$ lies in the span of $p_1(A), \dots,
p_d(A)$ for given polynomials $p_1,\dots,p_d\in\mathbb{Q}[x]$. The reduction
takes polynomial time, relies on standard linear algebra and is straightforward
to extend, \emph{mutatis mutandis}, in order to include linear inequalities on the coefficients which
witness membership of $A^nx$ in the target vector space.  Thus, we shall assume
that a problem instance of the Extended Orbit Problem is specified by matrices
$A\in\mathbb{Q}^{m\times m}$, $B\in(\ra)^{k\times d}$ and polynomials
$p_1,\dots, p_d\in\mathbb{Q}[x]$ such that $p_1(A),\dots,p_d(A)$ are linearly
independent, and we have to determine whether there exist
$n\in\mathbb{N}$ and $u=(u_1,\dots,u_d)\in\mathbb{Q}^d$ such that
\begin{equation}\label{eq: EOstatement}
A^n = u_1p_1(A) + \dots + u_dp_d(A)\mbox{ and } Bu \geq 0
\end{equation}

We now proceed to show a \emph{Master System} of equations, which is
equivalent to (\ref{eq: EOstatement}).  Let $f_A(x)$ be the minimal
polynomial of $A$ over $\mathbb{Q}$ and let $\alpha_1,\dots,\alpha_t$ be its
roots, that is, the eigenvalues of $A$.  These can be calculated in polynomial
time. Throughout this paper, for an eigenvalue $\alpha_i$ we will denote by
$\mathit{mul}(\alpha_i)$ the multiplicity of $\alpha_i$ as a root of the
minimal polynomial of the matrix.

Fix an exponent $n$ and coefficients $u_1,\dots,u_d$ and define the polynomials
$P(x)=\sum_{i=1}^du_ip_i(x)$ and $Q(x)=x^n$. It is easy to see that (\ref{eq: 
EOstatement}) is satisfied if and only if
\begin{equation}\label{eq: master}
Bu\geq0 \wedge 
P^{(j)}( \alpha_i) =Q^{(j)}( \alpha_i)
\end{equation}
for all $i\in\{1, \dots, t\}$, $j\in\{0,\dots,\mathit{mul}(\alpha_i)-1 \}$.
Indeed, $P-Q$ is zero at $A$ if and only if $f_A(x)$ divides $P-Q$, that is,
each $\alpha_i$ is a root of $P-Q$ with multiplicity at least
$\mathit{mul}(\alpha_i)$. This is equivalent to saying that each $\alpha_i$ is
a root of $P-Q$ and its first $\mathit{mul}(\alpha_i)-1$ derivatives.

Thus, in order to decide whether the problem instance is positive, it is
sufficient to solve the system of equations and inequalities (\ref{eq: master})
in the unknowns $n$ and $u_1,\dots,u_d$.  Each eigenvalue $\alpha_i$
contributes $\mathit{mul}(\alpha_i)$ equations which specify that $P(x)-Q(x)$ and
its first $\mathit{mul}(\alpha_i)-1$ derivatives all vanish at $\alpha_i$.

For example, if $f_A(x)$ has roots $\alpha_1$, $\alpha_2$, $\alpha_3$ with
multiplicities $\mathit{mul}(\alpha_i)=i$ and the target space is $span\left\{
p_1(A),p_2(A)\right\}$ then the system contains six equations, in addition to
the inequalities $Bu\geq0$:
\begin{eqnarray*}
\alpha_1^n & = & u_1p_1(\alpha_1)+u_2p_2(\alpha_1)\\
\alpha_2^n & = & u_1p_1(\alpha_2)+u_2p_2(\alpha_2)\\
n\alpha_2^{n-1} & = & u_1p_1'(\alpha_2)+u_2p_2'(\alpha_2)\\
\alpha_3^n & = & u_1p_1(\alpha_3)+u_2p_2(\alpha_3)\\
n\alpha_3^{n-1} & = & u_1p_1'(\alpha_3)+u_2p_2'(\alpha_3)\\
n(n-1)\alpha_3^{n-2} & = & u_1p_1''(\alpha_3)+u_2p_2''(\alpha_3)
\end{eqnarray*}
Notice also that we may assume without loss of generality that $0$ is not
an eigenvalue. Otherwise, its equations in the Master System
$0=u_1p_1^{(j)}(0)+\dots+u_dp_d^{(j)}(0)$ either yield a linear dependence
on $u_1,\dots,u_d$, allowing us to eliminate some $u_i$ and proceed inductively
by solving a lower-dimensional Master System, or are trivially satisfied by 
all $u_1,\dots,u_d$ and may be dismissed.

\subsection{Equivalence classes of $\sim$}

Next, we focus on the equivalence relation $\sim$ on the eigenvalues of the
input matrix defined by
\[\alpha\sim\beta\iff\alpha/\beta\mbox{ is a root of unity}\] 
The image of an equivalence class of $\sim$ under complex conjugation is also
an equivalence class of $\sim$. If a class is its own image under complex
conjugation, then it is called \emph{self-conjugate}. Classes which are not
self-conjugate are grouped into \emph{pairs of conjugate classes} which are
each other's image under complex conjugation.

If a class $\mathcal{C}$ is self-conjugate, then we can write it as
\[\mathcal{C}=\{\alpha\omega_1,\alpha\omega_2,\dots,\alpha\omega_s\}\] where
$\alpha$ is real algebraic and $\omega_1,\dots,\omega_s$ are roots of unity. This
representation is easily computable in polynomial time.  Similarly, if two
classes $\mathcal{C}_1,\mathcal{C}_2$ are each other's image under complex
conjugation, they can be written as
\[
\begin{array}{c}
\mathcal{C}_1=\left\{\alpha\omega_1,\alpha\omega_2,\dots
,\alpha\omega_s\right\}\\
\\
\mathcal{C}_2=\left\{\overline{\alpha\omega_1},\overline{\alpha\omega_2},\dots,
\overline{\alpha\omega_s}\right\}
\end{array}
\]
where $\alpha$ is algebraic and $\arg(\alpha)$ is not a rational
multiple of $2\pi$.  For an equivalence class $\mathcal{C}$ of $\sim$,
write $\mathit{Eq}(\mathcal{C}, j)$ for the set of $j$-th derivative
equations contributed to the Master System by eigenvalues in
$\mathcal{C}$. Define also the \emph{multiplicity} of $\mathcal{C}$ to
be the maximum multiplicity of an eigenvalue in $\mathcal{C}$.

In our work on the Orbit Problem \cite{orbitArxiv}, we analysed the equivalence
classes of $\sim$ in order to derive a bound on the exponent $n$.  We were able
to show that if $A$ has `sufficiently many' eigenvalues unrelated by $\sim$
then just the condition $A^n\in\spa\{p_1(A),\dots,p_d(A) \}$ on its own is
strong enough to bound the exponent, regardless of the linear inequalities
$Bu\geq0$ which the Extended Problem imposes on the coefficients
$u_1,\dots,u_d$. The following theorem will allow us to focus only on the cases
in which $\sim$ has `few' equivalence classes.
\begin{theorem}\label{thm: orbit}
Suppose we are given a problem instance $(A, B, p_1,\dots, p_d)$ with $d\leq3$ and
let $\sim$ be the relation on the eigenvalues of $A$ defined as above. Write
$\Vert I\Vert = \Vert A\Vert+\Vert p_1\Vert
+\dots+\Vert p_d\Vert$. Let $R$ be the sum of the multiplicities of the
equivalence classes of $\sim$. Then if $R\geq d+1$, then there exists an
effectively computable bound $N\in\ein{\Vert I\Vert}$ such that if
$A^n\in\spa(\{p_1(A),\dots,p_d(A)\})$, then $n\leq N$. Moreover, if
$d=1$, then $N\in\pin{\Vert I\Vert}$.
\end{theorem}

\subsection{Case analysis on the residue of $n$}\label{subsec: caseAnalysis}

Let $L$ be the least common multiple of all the orders of the ratios of
eigenvalues of $A$ which are roots of unity. Notice that $L\in\ein{\Vert
I\Vert}$. In the two- and three-dimensional Extended Orbit Problem, we will
perform a case analysis on the residue of $n$ modulo $L$. We will show that for
each fixed residue of $n$, we can either solve the problem instance directly or
derive an effective bound $N$ such that any witness $n$ to the problem instance
must be bounded above by $N$.  Since $L$ is at most exponentially large, it may
be expressed using at most polynomially many bits. Thus, when the relation
$\sim$ has too few equivalence classes for Theorem \ref{thm: orbit} to apply,
our polynomial-space algorithm can guess the residue of $n$ modulo $L$. This
greatly simplifies the Master System and either allows us to solve it outright
or to reduce it to an instance of the Simultaneous Positivity Problem.  

We now
consider what happens to the equations in $\mathit{Eq}(\mathcal{C}, j)$ for a
fixed residue of $n$ modulo $L$.  Let $\mathcal{C}=\{\alpha\omega_1, \dots,
\alpha\omega_s\}$ and for simplicity consider $\mathit{Eq}(\mathcal{C}, 0)$:
\begin{align*}
	(\alpha\omega_1)^n  = & \sum_{i=1}^d u_ip_i(\alpha\omega_1) \\
	\dots & \\
	(\alpha\omega_s)^n  = & \sum_{i=1}^d u_ip_i(\alpha\omega_s) \\
\end{align*}
This set of equations is equivalent to 
\begin{equation}\label{ven: collapsing}
\alpha^n 
= \sum_{i=1}^d u_i\frac{p_i(\alpha\omega_1)}{\omega_1^n}
= \dots
= \sum_{i=1}^d u_i\frac{p_i(\alpha\omega_s)}{\omega_s^n}
\end{equation}
For a fixed residue of $n$ modulo $L$, we see 
$\omega_1^n,\dots,\omega_s^n$ are also fixed, so
each $p_i(\alpha\omega_j)/\omega_j^n$ is easily computable.
Observe that (\ref{ven: collapsing}) is equivalent to the conjunction
of an equation with a linear system:
\begin{equation}\label{ven: collapsedEq}
\alpha^n = 
\sum_{i=1}^d u_i\frac{p_i(\alpha\omega_s)}{\omega_s^n}
\mbox{ and } B'u = 0
\end{equation}
where $B'$ is an $(s-1)\times k$ matrix over $\mathbb{A}$ defined by
\[ 
B'_{j,i} = \frac{p_i(\alpha\omega_j)}{\omega_j^n} - \frac{p_i(\alpha\omega_{j+1})}{\omega_{j+1}^n}
\]
Writing $\varphi_i$ for $p_i(\alpha\omega_s)/\omega_s^n$ and considering 
separately the real and imaginary parts of $B'u=0$, we see that 
(\ref{ven: collapsedEq}) is equivalent to
\[
\alpha^n = \varphi_1u_1 + \dots + \varphi_du_d
\mbox{ and }
B''u = 0
\]
where 
\[
B'' =\left[\begin{array}{c}
\re(B') \\
\im(B') \\
\end{array}\right]
\] 
is a $2(s-1)\times k$ matrix over $\ra$. However, $u$ lies in the nullspace of
$B''$ if and only if $u$ is orthogonal to the column space of $B''$. Thus, if
$B''$ has non-zero column rank, then $u_1,\dots,u_d$ must have a non-trivial
linear dependence $\psi_1u_1+\dots+\psi_du_d=0$ for effectively computable 
$\psi_1, \dots, \psi_d\in\ra$. Therefore, we can eliminate some coefficient 
$u_i$, replacing all of its occurrences in the Master System (\ref{eq: master}),
and proceed inductively to solve a Master System with dimension $d-1$. 
Therefore, we can assume that the column rank of $B''$ is zero, so the 
constraint $B''u=0$ is satisfied by all vectors $u$.

Thus, for this particular residue of $n$ modulo $L$, the equations
$\mathit{Eq}(\mathcal{C}, 0)$ are equivalent to the single equation 
$\alpha^n = \varphi_1u_1 + \dots + \varphi_du_d$.
Further, if the equivalence class $\mathcal{C}$ is self-conjugate, then
$\alpha\in\ra$, so we may replace each $\varphi_i$ with its real part 
and assume $\varphi_i\in\ra$.
Similarly, for $j>0$ and a fixed residue of $n$ modulo $L$, the equations 
$\mathit{Eq}(\mathcal{C}, j)$ reduce to the equivalent single equation
\[
n(n-1)\dots (n-j+1)\alpha^{n-j} = 
\sum_{i=1}^d u_i\frac{p_i^{(j)}(\alpha\omega_s)}{\omega_s^{n-j}}
\]

\subsection{One-dimensional case of Extended Orbit}\label{app:EO1}

In the one-dimensional Extended Orbit Problem, we have to decide whether there
exists some $n\in\mathbb{N}$ such that $A^n$ is a non-negative multiple of
$p_1(A)$. We show this problem is in $\mathbf{PTIME}$.

Begin by observing that if $0$ is an eigenvalue of $A$, then its equations in
the Master System are either satisfied for all values of $n$, or for no values
of $n$. In the former case, they can be discarded, whereas in the latter case,
the problem instance is immediately negative. We will now perform a case 
analysis on the number of equivalence classes of $\sim$.

\textbf{Two or more equivalence classes.}
When the relation $\sim$ has at least two equivalence classes, by Theorem
\ref{thm: orbit}, there exists a computable bound $N\in\pin{||I||}$ on the
exponent $n$. It suffices to try all $n\leq N$, which can be done in polynomial
time.

\textbf{One equivalence class, all roots simple.}
The second case is when $\sim$ has only one equivalence class and the 
eigenvalues $\alpha_1,\dots,\alpha_s$ of $A$ are all simple in the 
minimal polynomial of $A$. 
The Master System is then equivalent to
\begin{align*}
u_1 & = \frac{\alpha_1^n}{p_1(\alpha_1)} \geq 0 \\
\left(\frac{\alpha_i}{\alpha_j}\right)^n & = \frac{p_1(\alpha_i)}{p_1(\alpha_j)}
\mbox{ for all $i, j$} \\
\end{align*}
Since all ratios $\alpha_i/\alpha_j$ are roots of unity,
each equation $(\alpha_i/\alpha_j)^n=p_1(\alpha_i)/p_1(\alpha_j)$ is either
unsatisfiable, making the problem instance immediately negative, or
equivalent to some congruence in $n$. If all equations are satisfiable,
then $A^n
\in\spa\{p_1(A)\}$ holds if and only if 
$n\equiv t_1 \mbox{ mod } t_2$, where $t_1, t_2$ are effectively computable
natural numbers. Moreover, since $\sim$ has only one equivalence class,
it must necessarily be self-conjugate, so $\alpha_1 = |\alpha_1|\omega$
for some root of unity $\omega$ which can be calculated easily. Since
$u_1 = \re(\alpha_1^n)/\re(p_1(\alpha_1))$, we can compute what
the sign of $\re(\alpha_1^n)$ should be to ensure $u_1\geq0$, that is, whether
$\omega^n$ must be $1$ or $-1$ for $n$ to be a witness. This leads to
another congruence in $n$ which we put in conjunction with 
$n\equiv t_1 \mbox{ mod } t_2$. The problem instance is positive iff
the two congruences have a common solution.

\textbf{One equivalence class, some repeated roots.}
As in the previous case, we take the ratios of all pairs of equations
$\alpha_i^n = u_1p_1(\alpha_i)$ and $\alpha_j^n = u_1p_1(\alpha_j)$, giving 
\begin{equation}\label{ven:allpairs}
\left(\frac{\alpha_i}{\alpha_j}\right)^n  = \frac{p_1(\alpha_i)}{p_1(\alpha_j)}
\mbox{ for all $i, j$}
\end{equation}
Additionally,
for each repeated root $\alpha_i$, we take the ratios of its first and
second equation, of its second and third equation, and so on, obtaining
\begin{equation}\label{ven:allrep}
\frac{\alpha_i}{n-j}=\frac{p_1^{(j)}(\alpha_i)}{p_1^{(j+1)}(\alpha_i)}
\mbox{ for all $j\in\{0,\dots,\mathit{mul}(\alpha_i)-1\}$}
\end{equation}
If the equations (\ref{ven:allrep}) point to different values of $n$, then
the problem instance is negative. If they point to the same value of $n$, 
but $n$ does not satisfy the congruence resulting from (\ref{ven:allpairs}),
then the problem instance is negative. Otherwise, the problem instance is 
positive if and only if $u_1=\alpha_1^n/p_1(\alpha_1)$ is positive. The
relation $\sim$ has only one equivalence class, so it must be self-conjugate,
so $\alpha_1=|\alpha_1|\omega$ for some computable root of unity $\omega$.
It is easy to check the sign of $\omega^n$, so the decision method is
complete.

\subsection{Two-dimensional case of Extended Orbit}

Now suppose we have a problem instance $(A,B,p_1,p_2)$ and we have to 
determine whether there exist an exponent
$n\in\mathbb{N}$ and coefficients $u=(u_1,u_2)\in\mathbb{Q}^2$ such that
\[
A^n = u_1p_1(A) + u_2p_2(A) 
\mbox{ and } 
Bu\geq0
\]
We will perform a case analysis on the equivalence classes of $\sim$. By
Theorem \ref{thm: orbit}, if the sum of the multiplicities of the equivalence
classes of $\sim$ is at least $3$, then there exists an effective bound
$N\in\ein{\Vert I\Vert}$ on $n$ such that for $n>N$, mere membership of $A^n$
in $\spa( \{p_1(A),p_2(A)\}$ is impossible, regardless of the constraints on
the coefficients $u_1,u_2$.  Then an exponent $n\leq N$ can be chosen 
nondeterministically and verified using a $\mathbf{PosSLP}$ oracle. We 
consider the remaining cases.

\textbf{One simple equivalence class}. Suppose $\sim$ has only one equivalence 
class and its eigenvalues are all simple in the minimal polynomial of the matrix.
We proceed by case analysis on the residue of $n$, as in Section \ref{subsec: 
caseAnalysis}. For a fixed residue, the Master System reduces to
\begin{equation}\label{eq:singleClassEq}
\alpha^n = u_1\varphi_1 + u_2\varphi_2 \mbox{ and } Bu\geq0
\end{equation}
where $\varphi_1,\varphi_2\in\ra$. 
Fix the parity of $n$ and therefore assume $\alpha > 0$ by including its sign 
into $\varphi_1,\varphi_2$. Now observe that either all values of $n$ satisfy 
(\ref{eq:singleClassEq}), or no value of $n$ does.
Indeed, if $n$ is a witness with coefficients $(u_1,u_2)$,
then $n+1$ and $n-1$ are also witnesses, with coefficients 
$(u_1\alpha,u_2\alpha)$ and $(u_1/\alpha,u_2/\alpha)$, respectively.
Therefore, it suffices to try $n=0$. This leads to a conjunction of 
the equation $1=u_1\varphi_1+u_2\varphi_2$ with inequalities in 
$u_1,u_2$, which is easy to solve.

\textbf{Two simple equivalence classes.} Suppose that $\sim$ has two
equivalence classes and all eigenvalues are simple in the minimal polynomial
of the matrix. Proceed by case analysis on the residue of $n$ as before and
reduce the Master System to
\begin{equation}\label{eq: EO2,11}
\left[
\begin{array}{c}
	\alpha^n \\
	\beta^n \\
\end{array}
\right]
=
\left[
\begin{array}{cc}
	\varphi_1 & \varphi_2 \\
	\varphi_3 & \varphi_4 \\
\end{array}\right]
\left[
\begin{array}{c}
	u_1 \\
	u_2 \\
\end{array}
\right]\mbox{ and } Bu\geq0
\end{equation}
If the equivalence classes are both self-conjugate, then $\varphi_1,\dots,\varphi_4,
\alpha,\beta$ are all real algebraic, otherwise $\varphi_3 = \overline{\varphi_1}$, $\varphi_4= 
\overline{\varphi_2}$ and $\alpha=\overline{\beta}$. If the $2\times2$ matrix
in (\ref{eq: EO2,11}) is invertible, then premultiplying by its inverse yields 
\[
\left[
\begin{array}{c}
	u_1 \\
	u_2 \\
\end{array}
\right]
=
\left[
\begin{array}{cc}
	\psi_1 & \psi_2 \\
	\psi_3 & \psi_4 \\
\end{array}\right]
\left[
\begin{array}{c}
	\alpha^n \\
	\beta^n \\
\end{array}
\right]\mbox{ and } Bu\geq0
\]
where either $\psi_1,\dots,\psi_4$ are real, or $\psi_2=\overline{\psi_1}$ and
$\psi_4 = \overline{\psi_3}$. Now observe that $u_1, u_2$ satisfy a linear 
recurrence formula with characteristic equation $(x-\alpha)(x-\beta)=0$. Then
$Bu$
is a vector of linear recurrence sequences over $\ra$.
Each sequence $\ess_i(n)$ has order at most 2 and is given by
\[ \ess_i(n)=a_i\alpha^n + b_i\beta^n \]
so they all satisfy the same shared recurrence formula. Further, observe that
these sequences are non-degenerate, since $\alpha/\beta$ is not a root of
unity. Therefore, for this particular residue of $n$, the problem instance
reduces to Simultaneous Positivity for sequences of order at most 2.  
Finally, if the $2\times2$ matrix in 
(\ref{eq: EO2,11}) is singular, then there is a non-trivial linear combination of the rows
which equates to zero. Then the same nontrivial combination of $\alpha^n,\beta^n$
equals zero. A bound on $n$ follows from Theorem \ref{thm: skolem}.

\textbf{One repeated equivalence class.} The last remaining case is when there
is only one equivalence class of $\sim$ and it contains at least one eigenvalue
repeated in the minimal polynomial of $A$. This reduces to
Simultaneous Positivity in the same way as the previous case, but the
resulting recurrence sequences have characteristic equation 
$(x-\alpha)^2=0$ and are given by $\ess_i(n) = (a_i+b_in)\alpha^n$.

\subsection{Three-dimensional case of Extended Orbit}

Now we consider an instance of the Extended Orbit Problem with a 
three-dimensional target space. For given $(A,B,p_1,p_2,p_3)$,
we need to determine whether there exist $n\in\mathbb{N}$ and
$u=(u_1,u_2,u_3)\in\mathbb{Q}^3$ such that
\[
A^n = u_1p_1(A) + u_2p_2(A) + u_3p_3(A) \mbox{ and }
Bu\geq0
\]
The strategy is again to show an effective bound $N$ such that if there is a
witness $(n,u_1,u_2,u_3)$ to the problem instance, then $n<N$. By Theorem
\ref{thm: orbit}, we need only bound $n$ in the cases when the multiplicities
of the equivalence classes sum to at most $3$. 

\textbf{Three simple equivalence classes.} If there are exactly three classes,
each of multiplicity 1, one must necessarily be self-conjugate whereas the
other two can be either self-conjugate or each other's conjugates. Either way,
this case is analogous to the case of two simple equivalence classes
in the two-dimensional version.
After performing a case analysis on the residue of $n$, we obtain
\begin{equation}\label{eq: 3dEO,111}
	\left[\begin{array}{c}
		\alpha^n \\
		\beta^n \\
		\gamma^n
	\end{array}\right]
	=
	Tu\mbox{ and } Bu\geq0
\end{equation}
where $T$ is a $3\times3$ matrix over $\ra$.  If $T$ is invertible, then we
multiply both sides of (\ref{eq: 3dEO,111}) by $T^-1$ and see that $u_1,u_2,u_3$ are
linear recurrence sequences over $\ra$ with characteristic roots
$\alpha,\beta,\gamma$. Then the left-hand side of each linear inequality $Bu\geq0$
is also an LRS over $\ra$ and has order $3$.  Thus the problem
instance reduces to Simultaneous Positivity for order-3 sequences. On the other
hand, if $T$ is singular, then a linear combination of its rows is zero, so the
same linear combination of $\alpha^n,\beta^n,\gamma^n$ is also zero. Noting
that no two of $\alpha,\beta,\gamma$ are related by $\sim$, we obtain a bound 
on $n$ from Theorem \ref{thm: skolem}.

\textbf{Two classes, one simple and one repeated.} Next, suppose $\sim$ has two
equivalence classes, one of multiplicity $1$ and the other of multiplicity $2$.
This is analogous to the previous case.  For a fixed residue of $n$ modulo $L$,
the Master System is equivalent to 
\begin{equation}\label{eq: 3dEO,21}
	\left[\begin{array}{c}
		\alpha^n \\
		n\alpha^{n-1} \\
		\beta^n
	\end{array}\right]
	=
	Tu\mbox{ and } Bu\geq0
\end{equation}
where $T$ is a $3\times3$ matrix over $\ra$.  Now if $T$ is invertible, then we
multiply both sides of (\ref{eq: 3dEO,21}) by $T^{-1}$ and see that each of
$u_1,u_2,u_3$ is a linear recurrence sequence over $\ra$ with characteristic
equation $(x-\alpha)^2(x-\beta)=0$. Substituting into the homogeneous linear
inequalities $Bu\geq0$, we now have an instance of the Simultaneous
Positivity Problem for LRS of order 3 with a repeated characteristic root. If
$T$ is singular, then a linear combination of $\alpha^n$, $n\alpha^{n-1}$ and
$\beta^n$ must equal zero, so a bound on $n$ follows from Theorem \ref{thm:
 skolem}, because the ratio of $\alpha$ and $\beta$ is not a root of unity.

\textbf{One simple equivalence class.} Suppose now that $\sim$ has only one
equivalence class and it has multiplicity $1$. The situation is analogous to
the same case in the two-dimensional version.  We have to find
$n,u_1,u_2,u_3$ such that 
\begin{equation}\label{eq: 3dEO,1}
\alpha^n = u_1\varphi_1 + u_2\varphi_2 + u_3\varphi_3 \mbox{ and } Bu\geq0
\end{equation}
Since
everything is real, we observe that either all $n$ are witnesses to the problem
instance, or none are, so it suffices to consider $n=0$, reducing the problem
to a conjunction of the linear inequalities $Bu\geq0$ with the equation
$1 = u_1\varphi_1 + u_2\varphi_2 + u_3\varphi_3$.

\textbf{Two equivalence classes, both simple.} 
Let $\sim$ have two equivalence classes, both of multiplicity $1$.
For a fixed residue of $n$ modulo $L$, the Master System is equivalent
to
\[
\left[
\begin{array}{c}
	\alpha^n \\
	\beta^n \\
\end{array}
\right]
=
Tu \mbox{ and } Bu\geq0
\]
where $T$ is a $2\times3$ matrix. All the numbers involved are algebraic.
There are two possibilities: either $\alpha$, $\beta$ and $T$ are in $\ra$, or
$\alpha=\overline{\beta}$ and the second row of $T$ is the complex conjugate
of the first row.

The dimension of the column space of $T$ is $0$,
$1$ or $2$. If the dimension of the column space is $0$, then the Master System
is unsatisfiable, since $T$ maps everything to zero, whereas $\alpha^n$ and
$\beta^n$ cannot be zero. If the dimension of the column space of $T$ is
$1$, then it is spanned by a single vector $(t_1, t_2)$. If at least one of
$t_1, t_2$ is zero, then the System is unsatisfiable, because $\alpha,
\beta\neq0$. Otherwise, we can conclude that $(\alpha/\beta)^n = t_1/t_2$.
Since $\alpha/\beta$ is not a root of unity, a bound on $n$ which is polynomial
in $\Vert I\Vert$ follows by Theorem \ref{thm: skolem}.

Assume therefore that the dimension of the column space of $T$ is $2$. We consider the
real and the complex cases separately. First, suppose $T,\alpha,\beta$ are real.
Each of the inequalities $Bu\geq0$
specifies that $(u_1,u_2,u_3)$ lies in a halfspace $\mathcal{H}_i$ of $\mathbb{R}^3$.
The image of each $\mathcal{H}_i$ under $T$ can be the entire plane 
$\mathbb{R}^2$, a half-plane, a line, or a half-line. Each of these images is
easy to calculate in polynomial time. If for some $i$, the image 
$T\mathcal{H}_i$ is a line or a half-line, with defining vector 
$(t_1, t_2)$, then by the same reasoning as 
above, we see $(\alpha/\beta)^n = t_1/t_2$ and hence obtain a bound on $n$ from
Theorem \ref{thm: skolem}. Otherwise, we can assume that for all $i$,
$T\mathcal{H}_i$ is a halfplane $\{ (x, y) : A_ix + B_iy \geq 0 \}$ with
effectively computable $A_i,B_i\in\ra$. We have to determine whether there
exists $n\in\mathbb{N}$ such that $(\alpha^n,\beta^n)$ lies in the intersection
of these halfplanes. Noting that $A_i\alpha^n + B_i\beta^n$ as a function of
$n$ is a linear recurrence sequence over $\ra$ which has order $2$, we see that
this is now an instance of the Simultaneous Positivity Problem, so 
we are done by Theorem \ref{thm: simposDecidability}.

Suppose now that $\alpha$ and $\beta$ are complex conjugates, and the second
row of $T$ is the complex conjugate of the first. We may freely assume that
$|\alpha|=|\beta|=1$, since if the inequalities are satisfied by $(u_1,u_2,
u_3)$, then they are also satisfied by $(u_1/|\alpha|^n, u_2/|\alpha|^n,
u_3/|\alpha|^n)$. The image under $T$ of each halfspace $\mathcal{H}_i$ 
is a homogeneous cone in the complex plane. The same is true of the 
intersection $G=\cap_i T\mathcal{H}_i$ of these cones, which 
may in fact be computed explicitly. We need to determine whether there exists
$n\in\mathbb{N}$ such that $\alpha^n\in G$. Notice that $\{ \alpha^n : 
n\in\mathbb{N}\}$ is dense on the unit circle.
The intersection of the unit circle with $G$ could be a single point, or an
arc. 

Representing real and imaginary parts with variables over $\mathbb{R}$, we
construct a sentence $\tau$ in the first-order theory of the reals which states
that the intersection of $G$ with the unit circle is a single point.  We check the
validity of $\tau$, this can be done in polynomial time by Theorem \ref{thm: fo}.
If $\tau$ is false, then $G$ intersects the unit circle in an
arc, so by the density of $\alpha^n$ on the unit circle, the Master System is
satisfiable.  Otherwise, the intersection is a single point $z\in\mathbb{C}$.
Moreover, this point is effectively computable -- Renegar's algorithm hinges on
quantifier elimination, and will produce a quantifier free formula containing
exactly the minimal polynomials of $\mathit{Re}(z)$ and $\mathit{Im}(z)$. The
procedure is polynomial-time, so $\Vert z \Vert\in\pin{\Vert I \Vert}$. Now the
Master System is satisfiable if and only if there exists $n\in\mathbb{N}$ such
that $\alpha^n=z$. As $\alpha$ and $z$ both have descriptions polynomial in the
input size and $\alpha$ is not a root of unity, we see there exists a
polynomial bound on $n$ from Theorem \ref{thm: skolem}.

\textbf{One repeated equivalence class.} Finally, suppose
$\sim$ has a single equivalence class and its multiplicity is $2$. Then
for a fixed residue of $n$ modulo $L$, the Master System is equivalent to
\[
\left[
\begin{array}{c}
	\alpha^n \\
	n\alpha^{n-1} \\
\end{array}
\right]
=
Tu \mbox{ and } Bu\geq0
\]
where $\alpha$ and $T$ are both real algebraic. This is now handled 
analogously to the previous case for a real $T$ and reduces to Simultaneous
Positivity for LRS with characteristic equation $(x-\alpha)^2=0$.




\bibliography{paper}
\end{document}